\documentclass[10pt,a4paper]{article}

\usepackage{chicago}
\usepackage{amsmath}
\usepackage{amsthm}
\usepackage{amssymb}
\usepackage{graphicx}
\usepackage{multirow}
\usepackage{xspace}
\usepackage{url}
\usepackage{placeins}
\usepackage{rotating}
\usepackage{booktabs}
\usepackage{changepage}
\usepackage{array}

\newcommand{\sac}{SA-CCR\xspace}
\newcommand{\sacB}{RSA-CCR\xspace}
\newcommand{\sacLong}{standardised approach for measuring counterparty credit risk exposures\xspace}
\newcommand{\gmm}{Gaussian Market Model\xspace}
\newcommand{\hw}{\gmm}

\newcommand{\rfc}{request for comments\xspace}
\newcommand{\addon}{add-on\xspace}
\newcommand{\addons}{add-ons\xspace}
\newcommand{\Addon}{Add-on\xspace}
\newcommand{\Addons}{Add-ons\xspace}

\newcommand{\SD}{\ensuremath{\text{SD}}}
\newcommand{\ead}{\ensuremath{\text{EAD}}}
\newcommand{\rc}{\ensuremath{\text{RC}}}
\newcommand{\pfe}{\ensuremath{\text{PFE}}}
\newcommand{\eepe}{\ensuremath{\text{EEPE}}}
\newcommand{\sacm}{\ensuremath{\text{SA-CCR}}}
\newcommand{\thr}{\ensuremath{\text{TH}}}
\newcommand{\mta}{\ensuremath{\text{MTA}}}
\newcommand{\nica}{\ensuremath{\text{NICA}}}
\newcommand{\addAgg}{\ensuremath{\text{AddOn}^\text{aggregate}}}

\newcommand{\CF}{\ensuremath{\text{CF}}}

\newcommand{\addTrade}{\ensuremath{\text{AddOn}^\text{trade}}}
\newcommand{\addCashflow}{\ensuremath{\text{AddOn}^\text{cashflow}}}
\newcommand{\sfa}{\ensuremath{\text{SF}^{(a)}}}
\newcommand{\sfir}{\ensuremath{\text{SF}^{(\text{IR})}}}
\newcommand{\da}{\ensuremath{d^{(a)}}}
\newcommand{\dir}{\ensuremath{d^{(\text{IR})}}}
\newcommand{\mf}{\ensuremath{\text{MF}}}
\newcommand{\mfU}{\ensuremath{\text{MF}^\text{unmargined}}}
\newcommand{\mfM}{\ensuremath{\text{MF}^\text{margined}}}
\newcommand{\mpor}{\ensuremath{\text{MPOR}}}
\newcommand{\floor}{\ensuremath{\text{floor}}}

\newcommand{\ATM}{\ensuremath{\text{ATM}}}
\newcommand{\Float}{\ensuremath{\text{Float}}}
\newcommand{\Fixed}{\ensuremath{\text{Fixed}}}

\newcommand{\Regarb}{Dependence on economically-equivalent confirmations}
\newcommand{\regarb}{dependence on economically-equivalent confirmations}

\newcommand{\AS}{AS}

\makeatletter

\providecommand{\tabularnewline}{\\}
\numberwithin{table}{section}
\numberwithin{equation}{section}
\theoremstyle{plain}
\newtheorem{thm}{\protect\theoremname}[section]
\theoremstyle{plain}

\theoremstyle{remark}

\theoremstyle{definition}
\newtheorem{defn}[thm]{\protect\definitionname}
\theoremstyle{plain}
\newtheorem{prop}[thm]{\protect\propositionname}
\theoremstyle{definition}

\makeatother

\providecommand{\assumptionname}{Assumption}
\providecommand{\definitionname}{Definition}
\providecommand{\examplename}{Example}
\providecommand{\propositionname}{Proposition}
\providecommand{\remarkname}{Remark}
\providecommand{\theoremname}{Theorem}

\begin{document}
\title{Revising \sac}
\author{Mourad Berrahoui, Othmane Islah\footnote{Contacts: Mourad.Berrahoui@lloydsbanking.com, Othmane.Islah@lloydsbanking.com.  The views expressed in this presentation are the personal views of the author and do not necessarily reflect the views or policies of current or previous employers. Not guaranteed fit for any purpose. Use at your own risk.}, Chris Kenyon\footnote{Contact: Chris.Kenyon@mufgsecurities.com.  This paper is a personal view and does not represent the views of MUFG Securities EMA plc (“MUSE”).  This paper is not advice.  Certain information contained in this presentation has been obtained or derived from third party sources and such information is believed to be correct and reliable but has not been independently verified.  Furthermore the information may not be current due to, among other things, changes in the financial markets or economic environment.  No obligation is accepted to update any such information contained in this presentation.  MUSE shall not be liable in any manner whatsoever for any consequences or loss (including but not limited to any direct, indirect or consequential loss, loss of profits and damages) arising from any reliance on or usage of this presentation and accepts no legal responsibility to any party who directly or indirectly receives this material.}}
\date{08 April 2019\vskip5mm Version 2.10\vskip 5mm Accepted for publication in this form by {\it Risk}.  Version in {\it Risk} is considerably shortened to comply with page limits, and references this version for details.}
\maketitle

\begin{abstract}
We propose revising \sac to \sacB by making SA-CCR self-consistent and appropriately risk-sensitive by cashflow decomposition in a 3-Factor \gmm.
\end{abstract}

\newpage
\tableofcontents

%=========================================================
\section{Introduction}

The \rfc\ \cite{saccrProposed} offers an opportunity to address major issues with the Basel \sacLong \sac \cite{saccr,saccrFaq} and may inform other jurisdictions \cite{baselImplementation18}.  Major issues include: lack of self-consistency for linear trades; lack of appropriate risk sensitivity (zero positions can have material \addons; moneyness is ignored); \regarb.  Medium issues include: ambiguity of risk assignment (i.e. requirement of a single primary risk factor); and lack of clear extensibility.  The issues with \sac, and the point suggestions by other authors \cite{saccrComments}, highlight appropriate principles  on which to reconstruct \sac, namely appropriate risk-sensitivity (same exposure for same economics, positive exposure for non-zero risk), transparency, consistency, and extensibility.  The scope of this proposal is non-option interest rate products where these issues are most significant and our objective is a replacement for the current SA-CCR in this area which is consistent with the existing structure and principles of SA-CCR.  

SA-CCR was designed to better address margined and unmargined trades, reflect then-recent observed volatility levels, better represent netting, and be simple to implement \cite{saccr}.   Some recent regulations and industry equivalents use sensitivities  \cite{frtb2019,isda2018simm,baseliii2011} however these are focussed on market risk (including market risk of CVA) and thus have short horizons (days to weeks).  Using sensitivities within credit risk where there is a longer horizon (up to one year) would require including prediction of their behaviour.  This is a potentially useful avenue of investigation for a complete update of SA-CCR which we do not propose here.
 
We appreciate that any standard approach needs to strike a balance between simplicity and accuracy. Too much simplicity can lead to perverse situations such as the ones we describe here. But unsophisticated Banks that cannot implement a complex approach such as IMM still need to be adequately capitalised via a simple approach. The issues in \sac seem to be due to a slight excess
in simplification. Other Banks are probably aware of some of the issues we describe here and other issues that are not mentioned here. Also the questions in the consultation \cite{saccrProposed} clearly show that the regulator is fully aware of the limitations of the current approach but requests suggestions for improvements in special cases rather than an overhaul of the approach.

Here we propose here a global solution that revises the \sac instead of a quick fix such as those in \cite{saccrFaq}. We do not think it is too late to make the standard more consistent at a marginal implementation cost (where instead of computation of trade level add-ons, a cashflow decomposition of a trade is performed and calculations of cashflow level add-ons for non-option trades) and remaining within the spirit of the rules and close to the formulae. This decomposition approach is already applied in \cite{saccr,saccrProposed} for options. A start in the model-based direction is present in \cite{saccrFoundations} which we generalize here.

From an implementation perspective, the \sacB is accurate and fully consistent, with a marginal increase in complexity to apply cashflow decomposition to non-options as well as options. It is simple enough for small Banks because the inputs required are used elsewhere and typically available as XML outputs from booking systems (projected cashflow, discount factors, cashflow fixing date/payment date).

Even if not adopted as the new Standard, we recommend at the minimum for the \sacB to become a standard control tool in other areas (such as impact of trade confirmation dependency on Pillar 1 capital calculations, adequacy of Pillar 1 capital) without the need of a more complex approach (such as IMM) to assess the capital adequacy.

Longer proofs, additional details, and derivations can be found in the accompanying Technical appendix \cite{berrahoui2019techical}.

%Booking-insensitive also brings \sac into line with \mifid's emphasis on economic equivalence \cite{mifid2} CHECK. 

%=========================================================
\section{Issues with \sac}

\begin{table}
\begin{centering}
\begin{tabular}{lp{3.9cm}p{1.3cm}p{5.5cm}}
\toprule
Impact & Issue  & Trades affected & Source  \\
\midrule
major & lack of self-consistency & linear  &  trades not decomposed to cashflows  \\
major & lack of appropriate risk sensitivity & linear  & trades not decomposed to cashflows  \\
major & \regarb      & linear   &  trades not decomposed to cashflows  \\
medium & ambiguity of risk assignment & linear & trades not decomposed to cashflows, and no explicit model \\
medium & lack of clear extensibility & non-vanilla & no explicit model \\
medium & lack of parameter transparency & linear & parameter not given in model terms\\
\bottomrule
\end{tabular}
\caption{Issues with \sac.  Major issues are those that potentially give rise to materially incorrect capital magnitude.  Medium issues potentially give rise to material mis-assignment of risk netting, e.g. curve risk being assigned to basis risk.}
\label{t:issues}
\end{centering}
\end{table}

The main issues with \sac are for linear instruments that we can illustrate using interest rate swaps (IRS). Major and medium issues are given in Table \ref{t:issues}. Major issues are those that potentially give rise to materially incorrect capital magnitude.  Medium issues potentially give rise to material mis-assignment of risk netting, e.g. curve risk being assigned to basis risk

We demonstrate that the issues with \sac can be material for simple situations and then describe \sacB.

\subsection{Reminder on \sac structure}

\begin{table}
\begin{centering}
\begin{tabular}{ll}
\bf Symbol & \bf Meaning \\ 
\toprule
\ead &  exposure at default \\
$\alpha$ & regulatory multiplier, set at $1.4$ \\
\eepe &  effective expected positive exposure (this is the IMM-CCR EEPE) \\
\rc &  replacement cost  \\
\pfe & potential future exposure (over one year)  \\ 
$C$ & haircut value of net collateral held calculated using \nica\ methodology\\
\thr & threshold \\
\mta & minimum transfer amount \\
 \nica & net independent collateral amount \\   
 \floor & regulatory floors, e.g. set at 5\%\ in Equation \ref{e:f} \\ 
 $\mf_i$ & maturity factor for trade $i$  \\
 $\da_i$ & adjusted notional for trade $i$ w.r.t. asset class $a$ \\
 $\sfa_i$ & supervisory factor for trade $i$ w.r.t. asset class $a$ \\
 \mpor & margin period of risk \\
 \midrule
 $r_t$ & short rate\\
 $a$ & mean reversion speed\\
 $\phi_t$  & deterministic mean reversion target\\
 $\sigma$ & volatility\\
 $W_t$ & Brownian motion\\
 $P(t,T)$ & price at $t$ of zero coupon bond with maturity $T$\\
 \midrule
 $\delta_{\tau}$ & day count fraction of underlying money market index\\
\bottomrule
\end{tabular}
\caption{Notation.  Horizontal lines separate notation used in different sections.}
\label{t:notation}
\end{centering}
\end{table}

We summarise salient features of \sac here for convenience.  Notation is given in Table \ref{t:notation}.
\[
\ead = \alpha \eepe = \alpha(\rc + \pfe)
\]
where for margined trades
\[
\rc=\max(V-C,\thr+\mta-\nica,0)
\]
and for unmargined trades $\thr+\mta-\nica\equiv 0$ so $\rc=\max(V-C,0)$.  Margined trades are capped at the unmargined level to allow for large $\thr+\mta$.  PFE is given by:

\begin{align}
\pfe_\sacm :=& f\left( \frac{V-C}{\addAgg} \right)  \addAgg    \nonumber    \\
f(x) :=& \min\left\{1,\  \floor + (1-\floor)\exp\left(  \frac{x}{2(1-\floor)}  \right)      \right\}  \label{e:f}
\end{align}
There is no diversification across asset classes, so \addons simply sum to create an aggregate \addon.   Trade level \addons are defined as
\[
\addTrade_i = \delta_i \sfa_i \da_i \mf_i
\]
\sfa\ captures risk factor volatility.  \da, the adjusted notional incorporates a duration measure for interest rate and credit instruments.

We add some details for IR trades as they are used below:
\begin{align*}
\dir_i &= N_i \times \SD_i\\
\SD_i &=\frac{e^{-0.05 S_i} - e^{-0.05 E_i}}{0.05}
\end{align*}
where $\SD_i$ is the statuary duration for trade $i$, and  $S_i,\ E_i$ are the start and end dates, and $N_i$ is the (time-averaged) notional.  The maturity factor for margined and un-margined cases for a trade with maturity $M_i$, in years
\begin{align*}
\mfU_i :=& \sqrt{ \min(\max(M_i,\floor),1)  }\\
\mfM :=& 1.5 \sqrt{\mpor}
\end{align*}
when all time units are in years.  The floor (10 or 20 business days) depends on the number of trades, disputes, etc.  The supervisory factor for interest rates is $0.5\%$.

Aggregation is done across the hedging set, where the hedging set is defined by the primary risk factor.  Rules for identifying the primary risk factor have been proposed in some jurisdictions \cite{saccr2017eba}.

\subsection{Lack of self-consistency for linear trades}

Economically equivalent positions have different capital requirements.  For example a vanilla IRS or a set of FRAs with the same strike have different \addons, see Table \ref{t:swaps}.  Thus \sac\ is not self-consistent.

\begin{table}[h]
	\centering
	\begin{tabular}{lr}
\toprule
      Instrument &  SA-CCR add-on \\
\midrule
        ATM Swap &      3,934,693 \\
 FRA replication &      3,433,691 \\
     Split at 3Y &      3,654,794 \\
\bottomrule
\end{tabular}

	\caption{Equivalent derivatives positions with different \addons.  ATM swap is 10Y USD, 100M notional.  Split means a 3Y swap and a separate forward starting swap (starting at 3Y for 7Y) with the same fixed rate as the ATM 10Y swap.}
	\label{t:swaps}
\end{table}

\FloatBarrier
%------------------------------------------------------------------------------------------------------
\subsection{Lack of appropriate risk sensitivity}

We illustrate two aspects here:  economically zero positions can have material \addons; and moneyness is ignored.  Ignoring moneyness is probably the more important, for example moneyness produces systematic risk for existing positions when the general level of rates changes.

To illustrate the first case consider an IRS economically hedged with a set of FRAs.  There is a non-zero add-on despite zero risk.  See Table  \ref{t:zeroaddon}. 

\begin{table}[h]
	\centering
	\begin{tabular}{lr}
\toprule
                 Instrument &  SA-CCR add-on \\
\midrule
 ATM net of FRA replication &      1,646,936 \\
\bottomrule
\end{tabular}

	\caption{Material \addon\ for zero economic position.  Position is ATM swap 10Y USD 100M notional hedged by strip of FRAs.}
	\label{t:zeroaddon}
\end{table}

Non-ATM swaps have the same regulatory \addons but can have wildly different actual EEPE (even using very simple dynamics).   See Table \ref{t:noatm}

\sac ignores the moneyness of linear derivatives although non-ATM swaps can have wildly different EEPE even using very simple dynamics. Table \ref{t:noatm} shows the EEPE versus moneyness and direction asof 5th December 2018.  EEPE was calculated using a one factor Hull-White model calibrated to \sac (mean reversion 0.05; flat volatility 0.0189).  Note that although the Supervisory Factor is 50bps this translates directly into a Hull-White volatility of 189bps.
 
 \begin{table}[h]
	\centering
	\begin{tabular}{lrr}
\toprule
                   Instrument &       Pay &      Receive \\
\midrule
                   ATM+500bps & 5,778,184 &    3,898,079 \\
                   ATM+100bps & 4,271,792 &    3,707,917 \\
                          ATM & 3,903,699 &    3,668,881 \\
                   ATM-100bps & 3,539,968 &    3,634,569 \\
                   ATM-500bps & 2,166,837 &    3,577,668 \\
 Fwd Start ATM 1y-10y +500bps & 4,712,043 &    4,891,493 \\
 Fwd Start ATM 1y-10y +100bps & 3,932,090 &    3,931,231 \\
         Fwd Start ATM 1y-10y & 3,737,350 &    3,691,418 \\
 Fwd Start ATM 1y-10y -100bps & 3,542,763 &    3,451,756 \\
 Fwd Start ATM 1y-10y -500bps & 2,766,729 &    2,495,425 \\
\bottomrule
\end{tabular}

	\caption{EEPE versus moneyness and direction calculated using one factor Hull-White model calibrated to \sac mean reversion 0.05; flat volatility 0.0189.  Under \sac the top five trades have the same \addon, as to the bottom five.   ATM swap is 10Y USD.  All instruments have USD 100M notional}
	\label{t:noatm}
\end{table}

Thus \sac\  does not have appropriate risk sensitivity since economically zero positions can have material \addons, and moneyness is ignored which can exclude material risk differences.

%\FloatBarrier
%------------------------------------------------------------------------------------------------------
\subsection{\Regarb}
 
We can construct zero \addons for any interest rate swap with maturity $T\geq$1 year where $T$ and $2T$ are in the same maturity bucket as follows:
\begin{itemize}
	\item A receiver (rate $R$) amortising swap with notional
	\(3\frac{e^{0.05\ T}}{e^{0.05\ T} - 1}\text{\ N}\) between today and
	\(T\), then notional \(\frac{e^{0.05\ T}}{e^{0.05\ T} - 1}N\) between
	\(T\) and \(2T;\)
	\item A payer (rate \(R\)) swap with maturity \(2T\) and notional
	\(2\frac{e^{0.05\ T}}{e^{0.05\ T} - 1}N\);
	\item A forward start (starting at \(T\) and maturity at \(2T\) ) receiver
	(rate R) swap with notional \(\frac{e^{0.05\ T}}{e^{0.05\ T} - 1}N;\)
	\item A payer swap (rate \(R\)) with maturity \(T\) and notional
	\(\frac{1}{e^{0.05\ T} - 1}N\) .
\end{itemize}
The adjusted notional for the first two swaps is the same, $2\frac{e^{0.05\ T}}{e^{0.05\ T} - 1}\text{\ N}$, and they are in opposite directions, so their net addon is zero.  \mf\ is the same for the third and fourth trade, 1 (as $T>1$), as is \sfir, both $0.05$.  The sum of direction and adjusted delta is:
\begin{align*}
\sum_{i=3,4} \delta_i\dir_i
=&  \frac{e^{0.05\ T}}{e^{0.05\ T} - 1}N \frac{e^{-0.05 T} - e^{-0.05 \times 2 T}}{0.05}
- \frac{1}{e^{0.05\ T} - 1}N  \frac{1  - e^{-0.05 T}}{0.05}\\
%=& \frac{1}{e^{0.05\ T} - 1}N  \left(  \frac{1  - e^{-0.05 T}}{0.05}   -  \frac{1  - e^{-0.05 T}}{0.05}               \right)\\
=&0
\end{align*}
So the net \addon\ for the third and fourth trades is also zero.  Now the net economics over all four trades is zero for $T$ to $2T$, but is a net receiver swap from $0$ to $T$ with flat notional $N(e^{0.05T} - 1)/(e^{0.05T} - 1) = N$.

Thus \sac has  \regarb.

\subsection{Medium issues}

There is only one primary risk factor for each trade.  This will mis-allocate risk netting where there are several key risk drivers, e.g. long term FX.  The lack of clarity that this produces is evident as some jurisdictions are creating additional rules to resolve the ambiguity of the single decision, e.g. \cite{saccr2017eba}.  Of course this does not solve the mis-allocation of risk netting.  As we will show below, cashflow decomposition used with an explicit model automatically produces risks in their appropriate hedging sets and buckets.  

There is a lack of transparency of the meaning of regulatory parameters in \sac.  This could be simply resolved if there was an explicit model.  For example although the interest rate Supervisory Factor for volatility is stated as 50bps this translates directly into a \hw volatility of 189bps.  Again, using an explicit model base automatically provides transparency provided the model is simple.

\FloatBarrier
%=========================================================
\section{\sacB: cashflow and model-based \sac}

To solve the issues with \sac\ we propose a cashflow and model-based \sac, which we call Revised \sac (\sacB).  We pick the model by identifying it from the \sac standard itself using the link between trade level \addons and the volatility of the present value of the trade in \cite{saccrFoundations}.  Basing \sacB on cashflows automatically solves the lack of self-consistency, lack of appropriate risk-sensitivity, \regarb.  Using a model assigns risk unambiguously, provides clear extensibility and provides parameter transparency.

%------------------------------------------------------------------------------------------------------
\subsection{Model identification}

Similarly to  \cite{saccrFoundations}, we can define the theoretical add-on at a horizon $T$ of a trade with value process $V_{t}$ as the following average expected positive exposure:
\begin{equation*}
\textrm{AddOn}_{V}(T)=\frac{1}{T}\int_{0}^{T}\mathbb{E}^{\mathbb{}}\left[\left(V(t)-V(0)\right)^{+}\right]dt 
\end{equation*}
Assuming $dV(t)= \sigma dW(t)$, then we obtain the trade level theoretical add-on at horizon
$T$  at :
\begin{equation}
\textrm{AddOn}_{V}(T)=\frac{2}{3}\sigma\sqrt{\frac{T}{2\pi}}\label{eq:TradeVolAddon}
\end{equation}
Recall the \sac trade add-ons for IR trades:
\begin{equation*}  
\addTrade=\delta  N\cdot SF^{IR} \frac{e^{-0.05 S} - e^{-0.05 E}}{0.05} \mf  
\end{equation*}
Taking a payer swap longer than one year means that $\mf=1$, and $\delta=1$. The supervisory factor $SF^{IR}=0.005$  for interest rates, so we define:
\[
\sigma_\text{HW} := \frac{3}{2}\sqrt{2\pi} SF^{IR}
\]
Consider a one factor Hull-White model  with short rate, $r(t)$ and mean reversion parameter  $a=0.05$ :
\begin{equation}
dr_t=a (\phi_t-r_t)  dt+\sigma_{HW}  dW_t             \label{e:HW}
\end{equation}
For an ATM swap  with start date $T_{s}$ and end date $T_{e}$ notional $N$, we show in Theorem \ref{th:identity} that the volatility of the present value of the swap (assuming a flat discount curve with zero rate) is:
\begin{equation*}
\sigma_{\ATM}(0)=N\frac{\sigma_{HW}}{a}(e^{-a T_{s}}-e^{-a T_{e} })
\end{equation*}
Given the relation (\ref{eq:TradeVolAddon}) between theoretical add-on and trade volatility, we obtain a model based add-on for the ATM swap, identical to the SA-CCR add-on :
\begin{align*}  
\addTrade&= N\cdot \frac{2\sigma_{HW}}{3\sqrt{2\pi}} \frac{e^{-a T_{s}} - e^{-a T_{e}}}{a}  \\
&= N\cdot  SF^{IR} \frac{e^{-a T_{s}} - e^{-a T_{e}}}{a}
\end{align*}
Thus \sac and \cite{saccrFoundations} themselves implicitly identify an appropriate model. Given that  \sac and \cite{saccrFoundations} can be interpreted as using a 1-Factor Hull-White model we extend this to a 3-factor model to take into account the correlation structure of the different maturity zones in Section \ref{s:3hw} below.   Correlations between \addons are mapped to correlations between zero-coupon bonds.

%For rates risk the supervisory factor has $\sigma(0)$ of 50 basis points so the equivalent $\sigma_\HW$ is 189 basis points.

%------------------------------------------------------------------------------------------------------
\subsection{Identity of \sac and 1-Factor Hull-White for single maturity bucket}

Here we provide the details to establish a formal correspondence of ATM Swap add-ons between \sac and HW 1 factor.    This correspondence being established, the regulatory parameters corresponding to model parameters are seen to be as in the previous section.  This process also shows how to extend \sac consistently to non-ATM trades.

A one factor Hull-White model \cite{hull1990pricing} has short rate, $r(t)$, dynamics given by Equation \ref{e:HW} above, so the dynamics of a zero coupon bond $P(t,T)$ price are:
\begin{equation}
dP(t,T)=r_t  P(t,T)  dt-\frac{\sigma}{a} (1-e^{-a(T-t)})  P(t,T)  dW_{t}     \label{e:ZCB}
\end{equation}
Consider a payer, i.e. pay-fixed rate $R$,notional $N$, swap price $V$ at $t$
\begin{align}
V(t) & =V_{float}(t)-V_{fixed}(t) \nonumber\\
 & =N\sum_{i=s+1}^{e}P(t,T_{i}) \delta_{i}  L(t,T_{i-1},T_{i})-N\sum_{i=s+1}^{e}P(t,T_{i}) \delta_{i}  R \nonumber
\end{align}
Now we have the identity between \sac and a 1-Factor Hull-White model in Theorem \ref{th:identity} below.

\begin{thm}
\label{th:identity} With zero bond dynamics
 given by Equation \ref{e:ZCB}, if $V$ is the value process
of a forward starting payer swap, then the instantaneous volatility of $V$ for $t\leq T_{s}$
can be decomposed into three contributions :
\[
\sigma_{V}(t)=\sigma_{\ATM}(t)+\sigma_{\Float}(t)+\sigma_{\Fixed}(t)
\]
Where :
\begin{align*}
\sigma_{\ATM}(t)= & N\sum_{i=s+1}^{e}P(t,T_{i-1})\frac{\sigma}{a}(e^{-a(T_{i-1}-t)}-e^{-a(T_{i}-t)})\\
\sigma_{\Float}(t)= & -N\sum_{i=s+1}^{e}P(t,T_{i})\frac{\sigma}{a}(1-e^{-a(T_{i-1}-t)})\delta_{i}L(t,T_{i-1},T_{i})\\
\sigma_{\Fixed}(t)= & N\sum_{i=s+1}^{e}P(t,T_{i})\frac{\sigma}{a}(1-e^{-a(T_{i}-t)})\delta_{i}R
\end{align*}
If at time $t$, the swap is ATM, taking the standard weight-freezing assumption (see proof) we have:
\begin{equation*}
\sigma_{V}(t) = N\sum_{i=s+1}^{e}P(t,T_{i})\frac{\sigma}{a}(e^{-a(T_{i-1}-t)}-e^{-a(T_{i}-t)})=\sigma_{\ATM}(t)
\end{equation*}
Moreover if the yield curve is flat and equal to zero, then the instantaneous
volatility of an ATM swap in one maturity bucket is given exactly by the \sac regulatory
formula i.e.:
\begin{equation}
\sigma_{\ATM}(t)=N\frac{\sigma}{a}(e^{-a(T_{s}-t)}-e^{-a(T_{e}-t)})\nonumber
\end{equation}
\label{th:1f}
\end{thm}
See Technical appendix \cite{berrahoui2019techical} for proof.

The swap volatility can be decomposed into two contributions: 
\begin{itemize}
\item volatility of the floating rate index itself: $\sigma_{ATM}(t)$
\item volatility of the present value of the cash-flows: $\sigma_{float}(t)+\sigma_{fixed}(t)$
\end{itemize}
Since for \sac only exposures below one year are relevant we freeze the volatility $\sigma_{V}(t)$ and use its
initial value $\sigma_{V}(0)$ so
\begin{align*}
\sigma_{V}(t) & \approx N \sum_{i=s+1}^{e}P(0,T_{i-1}) \frac{\sigma}{a} (e^{-aT_{i-1}}-e^{-aT_{i})})\\
 & +N \sum_{i=s+1}^{e}\delta_{i}P(0,T_{i}) \frac{\sigma}{a} (1-e^{-aT_{i}})\left(R-L(0,T_{i-1},T_{i})\right)
\end{align*}

\FloatBarrier
%------------------------------------------------------------------------------------------------------
\subsection{Matching inter-bucket correlations: 3-factor \hw}
\label{s:3hw}

\begin{table}
\centering
\begin{tabular}{cccc}
%\hline 
 & $T<1$ & $1\leq T<5$ & $5 \leq T$\tabularnewline
\hline 
%\hline 
$T<1$ & 1 & $\rho_{1}$ & $\rho_{2}$\tabularnewline
%\hline 
$1\leq T<5$ & $\rho_{1}$ & 1 & $\rho_{1}$\tabularnewline
%\hline 
$5\leq T$ & $\rho_{2}$ & $\rho_{1}$ & 1\tabularnewline
\hline 
\end{tabular}
\caption{\label{t:corr}Correlation structure of \sac 
interest rates add-ons by instrument maturity $T$, where \hbox{$\rho_{1}=0.7$} and $\rho_{2}=0.3$.}
\end{table}

Theorem \ref{th:identity} shows that the regulatory add-on formula for an at-the-money swap can be recovered from a 1-factor Hull-White model.  To match the inter-bucket correlations of \sac, Table \ref{t:corr}, we naturally move to a 3-factor \gmm:
\begin{itemize}
	\item We use maturity buckets defined by \sac: $M_{1}=[0,1)$, $M_{2}=[1,5)$, $M_{3}=[5,\infty)$
	\item  We also define  $M(t)=\sum_{i}1_{t\in M_{i}}M_{i}$ a map from positive numbers
	to the set of maturity buckets
	\item Next we define the correlated Brownian motions $Z_{t}^{M_{k}}$, $k=1,2,3$  with SA-CCR correlation structure i.e. 
	$dZ_{t}^{M_{1}}dZ_{t}^{M_{2}}=dZ_{t}^{M_{2}}dZ_{t}^{M_{3}}=0.7$
	and:
	$ dZ_{t}^{M_{1}}dZ_{t}^{M_{3}}=0.3$
\end{itemize}

 We define the following 3-factor \gmm extension of the Hull-White model that allows to recover the inter-bucket correlations for zero coupon bonds (with expiries $T$ in a discrete set ${T_{1},..,T_{N}}$ covering all business days up to horizon $T_{N}$) i.e :
\[
dP(t,T)=r_{t}P(t,T)\cdot dt-\frac{\sigma}{a}(1-e^{-a\cdot(T-t)})\cdot P(t,T)\cdot dZ_{t}^{M(T)}
\]

Where $r_{t}$ is the risk free rate (define as the continuously compounded rate of the shortest maturity zero-coupon bond).
Now we have a model identified from \sac we can move on to the second key element of \sacB, cashflow decomposition for linear products.

\subsection{Cashflow decomposition for linear products}

We first define the scope, concentrating on common linear product cashflows (which we call elementary cashflows) in Definition \ref{d:cf} and then provide Theorem \ref{thm:elementary cf addons} that gives the \addons for each of the different cashflow types.  In Section \ref{s:float}  we give an example derivation of the entries in Table 
\ref{t:CFaddon} for a Floating cashflow (with deterministic and stochastic basis).  Full derivation are in the accompanying Technical appendix \cite{berrahoui2019techical}.   In Section \ref{s:numerical} we compare the performance of \sac and \sacB.

From the examples here and in \cite{berrahoui2019techical}, the \addons for other (less common) cashflows in linear products can be derived in a similar manner.

\begin{defn}
	An elementary cashflow with payment date at $T$ and projected
	value $CF(T)$ (viewed from $0$) that is either the projected value
	of:
	\label{d:cf}
\end{defn}

\begin{itemize}
	\item a fixed cashflow $CF(T)=N$;
	\item a standard interest rates floating cashflow (as defined in point 2)
	, with notional $N$, money market index $L(t,t+\delta(\tau))$ tenor
	$\tau$ , fixing date $0<T_{f}\leq T$ defined as :
	\[
	\delta_{\tau} \  N \  L(T_{f},T_{f}+\delta_{\tau})
	\]
	\item a standard CMS cashflow with notional $N$ and swap rate $S_{\tau}(T_{s},T_{e})$
	fixing at $0<T_{s}\leq T$, with underlying money market tenor $\tau$,
	maturity at $T_{e}$, swap tenor $\delta(T_{s},T_{e})$, defined as
	:
\end{itemize}
\[
\delta_{\tau} \  N \  S_{\tau}(T_{s},T_{e})
\]

\begin{itemize}
	\item a standard inflation floating cashflow, with notional $N,$defined
	on inflation index $I(t),$with initial observation date $T_{s}$
	(can be in the past ), next observation date $\tau_{I}$ and final
	observation date $T_{e}\leq T$:
	\[
	N \ \left(\frac{I(T_{e})}{I(T_{s})}-1\right)
	\]
	\item a standard inflation compound cashflow with notional $N,$with initial
	observation date $T_{s}$ (can be in the past ), next observation
	date $\tau_{I}$ and final observation $T_{e}\leq T$:
	\[
	N \ \frac{I(T_{e})}{I(T_{s})}
	\]
	\item a standard interest rates compound cashflow , with start of compounding
	at $T_{s}$(cn be in the past i.e. , next reset date $\tau_{C}$ ,
	end of last compounding $T_{e}$, underlying index $L(t,t+\delta_{\tau})$,
	with tenor $\tau$. This cashflow is given by : 
	\begin{align*}
	N \ \prod_{k=s}^{e-1}(1+\delta_{\tau} \  L(T_{k},T_{k}+\tau)
	\end{align*}
\end{itemize}

\begin{thm}
	\label{thm:elementary cf addons}All elementary cashflows
	received (respectively paid) at time $T$, should have an add-on contribution using the supervisory formula
	\[
	\addCashflow = \delta \sfir \dir \mf  
	\]
	 to the maturity bucket corresponding
	to the payment date itself $T$ with follwing inputs to use in the
	supervisory formula:
	\begin{itemize}
		\item a delta equal to $-1$ (respectively $1$ ) ,
		\item start date of $0$ and end date of $T$ to compute the supervisory
		duration
		\item an adjusted notional equal to present value of the cashflow
	\end{itemize}
	Moreover, all non-fixed standard cashflows will have other add-ons
	contributions due to the volatility of their respective underlying
	indices, calculated using the supervisory formulae (with inputs as
	per Table \ref{t:CFaddon}) and allocated to the appropriate maturity
	buckets and hedging sets (as stated in Table \ref{t:CFaddon} for received cash-flows).
	
\begin{table}
\begin{adjustwidth}{-3.3cm}{-3.3cm}
	\begin{tabular}{|>{\raggedright}m{1.76cm}|>{\raggedright}m{1.55cm}|>{\centering}p{2.33cm}|>{\centering}p{4.3cm}|>{\centering}p{2.65cm}|>{\centering}p{2.68cm}|r@{\extracolsep{0pt}.}l|}
		\hline 
		Cashflow & Hedging set & Maturity bucket & Effective Notional & Start & End & \multicolumn{2}{c|}{Delta }\tabularnewline
		\hline 
		\hline 
		\multirow{2}{1.76cm}{Floating} & \multirow{2}{1.55cm}{Rates, Basis } & $M(T_{f}$) & $(N+CF(T)) \  P(0,T)$ & $0$ & $T_{f}$ & \multicolumn{2}{c|}{$-1$}\tabularnewline
		\cline{3-8} 
		&  & $M(T_{f}+\tau$) & $(N+CF(T)) \  P(0,T)$ & $0$ & $T_{f}+\tau$ & \multicolumn{2}{c|}{$+1$}\tabularnewline
		\hline 
		\multirow{2}{1.76cm}{Compound} & \multirow{2}{1.55cm}{Rates , Basis} & $M(\max(\tau_{c},T_{s}))$ & $CF(T) \  P(0,T)$ & 0 & $\max(\tau_{c},T_{s})$ & \multicolumn{2}{c|}{$-1$}\tabularnewline
		\cline{3-8} 
		&  & $M(T_{e})$ & $CF(T) \  P(0,T)$ & $0$ & $T_{e}$ & \multicolumn{2}{c|}{$+1$}\tabularnewline
		\hline 
		\multirow{3}{1.76cm}{CMS } & \multirow{3}{1.55cm}{Rates, Basis } & $M_{1}$ & $\frac{\delta_{\tau}}{\delta(T_{s},T_{e})}(N+CF(T)) \  P(0,T)$ & $\min(1,T_{s})$ & $\min(1,T_{e})$ & \multicolumn{2}{c|}{$-1$}\tabularnewline
		\cline{3-8} 
		&  & $M_{2}$ & $\frac{\delta_{\tau}}{\delta(T_{s},T_{e})}(N+CF(T)) \  P(0,T)$ & $\min(\max(1,T_{s}),5)$ & $\min(\max(1,T_{e}),5)$ & \multicolumn{2}{c|}{$+1$}\tabularnewline
		\cline{3-8} 
		&  & $M_{3}$ & $\frac{\delta_{\tau}}{\delta(T_{s},T_{e})}(N+CF(T)) \  P(0,T)$ & $\max(5,T_{s})$ & $\max(5,T_{e})$ & \multicolumn{2}{c|}{$+1$}\tabularnewline
		\hline 
		\multirow{2}{1.76cm}{Inflation Floating} & \multirow{2}{1.55cm}{Inflation} & $M(\max(\tau_{I},T_{s}))$ & $(N+CF(T)) \  P(0,T)$ & $0$ & $M(\max(\tau_{I},T_{s}))$ & \multicolumn{2}{c|}{$-1$}\tabularnewline
		\cline{3-8} 
		&  & $T_{e}$ & $(N+CF(T)) \  P(0,T)$ & $0$ & $T_{e}$ & \multicolumn{2}{c|}{$+1$}\tabularnewline
		\hline 
		\multirow{2}{1.76cm}{Inflation Compound} & \multirow{2}{1.55cm}{Inflation} & $M(\max(\tau_{I},T_{s}))$ & $CF(T) \  P(0,T)$ & $0$ & $M(\max(\tau_{I},T_{s}))$ & \multicolumn{2}{c|}{$-1$}\tabularnewline
		\cline{3-8} 
		&  & $T_{e}$ & $CF(T) \  P(0,T)$ & $0$ & $T_{e}$ & \multicolumn{2}{c|}{$+1$}\tabularnewline
		\hline 
	\end{tabular}
\caption{\Addons components for the elementary cashflows listed in Definition \ref{d:cf}.  Hedging sets for rates cashflows are Rates and Basis per currency.  Hedging sets for inflation cashflows are inflation per currency.  Derivations of these values follow the Floating  example in Section \ref{s:float} and details are given in the accompanying Technical appendix, including Floating and CMS examples in full. \label{t:CFaddon}}
\end{adjustwidth}
\end{table}
	With $CF(T)$ is the projected cash-flow, the map $M(t)$ that associates
	to $t$ its risk bucket and $P(0,T)$ the discount factor. Basis risk
	add-ons are allocated to the money market index $\tau$ vs discount
	hedging set.
	
	Setting $P(0,T)=1$ would result in simplified add-ons in line with
	current regulatory stipulations 
	\end{thm}

In \sacB \addon contributions for linear products are the aggregate add-on of the \addons for each cashflow .  This removes the \sac issues above: lack of self-consistency; lack of appropriate risk sensitivity; \regarb.  We formalize \addon calculations via cashflow decomposition in the definitions below.

\begin{defn}
A linear product has no optionality and is made of elementary cash-flows
$\CF_{i}$ with payment dates $T_{i}$ ($i=1,..,n)$ and notional $N_{i}$.
\end{defn}

\begin{defn}
The total \addon of a portfolio of linear products is obtained from
the aggregation of the individual \addons of each cash-flow
that makes up the linear products. \Addon aggregation
is performed according to \sac.
\end{defn}

\begin{prop}
All portfolios of linear products resulting in the same net cash-flows have the same aggregate \addon under \sacB.\label{p:decomp}
\end{prop}
\begin{proof}
By construction the \addons depend only on the net cash-flows.
\end{proof}
Propostion \ref{p:decomp} is the property we are after: this ensure appropriate risk sensitivity (no risk for no economic position and vice versa) and no \regarb.

%=========================================================
\section{Example \addon  derivation for a xibor coupon}
\label{s:float}

Here we provide an example of how to derive the entries in Table \ref{t:CFaddon} for a floating coupon for deterministic and stochastic index-vs-discounting basis.

The floating rate cash-flow is based on a money market index with fixing
at $T_{f}$ , tenor $\tau$ (less than $1$ year) and payment
at $T\geq\tau$ , with notional $N$. 

 We ignore convexity so this setup includes timing mismatches
such as a Libor rate paid in advance.

The forward value of the index is 
\[
\delta_{\tau}L(t,T_{f},T_{f}+\tau)=P_{f}(t,T_{f})/P_{f}(t,T_{f}+\tau)-1
\]
We define the index-vs-discounting basis $P_b$ via
\[
P_{f}(t,s)=P_{b}(t,s)  P(t,s)
\]
Ignoring convexity the coupon value is
\[
V(t)=N  P(t,T)(P_{f}(t,T_{f})/P_{f}(t,T_{f}+\tau)-1)
\]
Applying Ito's Lemma to $V(t)$ we see that the instantaneous volatility of $V(t)$ has three
components, one  from the floating rate index volatility, another
from the cashflow's present value and a third  from
the basis risk volatility:
\begin{align*}
dV_{t} & =(..)dt+N [\delta_{\tau}L(t,T_{f},T_{f}+\tau)  dP(t,T)\\
 & +P(t,T)P_{b}(t,T_{f})/P_{f}(t,T_{f}+\tau)  dP(t,T_{f})+P(t,T)  P_{f}(t,T_{f})/P_{b}(t,T_{f}+\tau)  d(1/P(t,T_{f}+\tau))\\
 & +P(t,T)  P(t,T_{f})/P_{f}(t,T_{f}+\tau)  dP_{b}(t,T_{f})+P(t,T)  P_{f}(t,T_{f})/P(t,T_{f}+\tau)  d(1/P_{b}(t,T_{f}+\tau))]
\end{align*}
No $dP_*dP_*$ terms are present above as these result in drift ($dt$) contributions not volatility ($dW_*$) contributions.

\subsection{Floating coupon: deterministic basis}
Assume first that the basis curve is deterministic, then we have for
$P(t,T)$, $P(t,T_{f})$ and $P(t,T_{f}+\tau)$ the following dynamics from the 3-factor \hw:
\begin{align*}
dP(t,T)=&(..)  dt-\frac{\sigma}{a}(1-e^{-a(T-t)})P(t,T)dZ_{t}^{M(T)}  \\
dP(t,T_{f})=&(..).dt-\frac{\sigma}{a}(1-e^{-a(T-t)})P(t,T)dZ_{t}^{M(T_{f})} \\
d(\frac{1}{P(t,T_{f}+\tau)})=&(..).dt+\frac{\sigma}{a}(1-e^{-a(T_{f}+\tau-t)}) \frac{1}{P(t,T_{f}+\tau)}dZ_{t}^{M(T_{f}+\tau)}
\end{align*}
Hence
\begin{align}
dV_{t} & =r_{t}V_{t}dt-\frac{\sigma}{a}(1-e^{-a(T-t)})  P(t,T)  N \delta_{\tau}L(t,T_{f},T_{f}+\tau)dZ_{t}^{M(T)}\nonumber \\
 & -\frac{\sigma}{a}(1-e^{-a(T_{f}-t)})  P(t,T)  N (1+\delta_{\tau}L(t,T_{f},T_{f}+\tau))  dZ_{t}^{M(T_{f})}  \label{e:flt}\\
 & +\frac{\sigma}{a}(1-e^{-a(T_{f}+\tau-t)})  P(t,T)  N (1+\delta_{\tau}L(t,T_{f},T_{f}+\tau))  dZ_{t}^{M(T_{f}+\tau)}\nonumber 
\end{align}
From Equation \ref{e:flt} we see that the floating rate cash-flow \addons bucket contributions are:
\begin{enumerate}
\item contribution from the present value of the cash-flow:
\begin{itemize}
\item Maturity bucket $M(T)$
\item Notional: $N  P(0,T) \tau  L(0,T_{f},T_{f}+\tau)$
\item Duration: $\frac{1}{a} (1-e^{-aT})$
\item Delta : $-1$
\end{itemize}
\item index volatility: fixing date contribution
\begin{itemize}
\item Maturity bucket $M(T_{f})$
\item Effective Notional: $N  P(0,T) (1+\delta_{\tau}  L(0,T_{f},T_{f}+\tau))$
\item Duration: $\frac{1}{a} (1-e^{-aT_{f}})$
\item Delta : $-1$
\end{itemize}
\item index volatility: payment date contribution
\begin{itemize}
\item Maturity bucket $M(T_{f}+\tau)$
\item Effective Notional: $N  P(0,T) (1+\tau  L(0,T_{f},T_{f}+\tau))$
\item Supervisory Duration: $\frac{1}{a} (1-e^{-a (T_{f}+\tau)})$
\item Delta : $1$
\end{itemize}
\end{enumerate}

\subsection{Floating coupon: stochastic basis}
If we assume a stochastic basis spread curve (some implementation scenarios for Libor decommissioning fallback may make this irrelevant) as
\[
dP_{b}(t,T_{f})=(..).dt-\frac{\sigma_{b}}{a}(1-e^{-a(T_{f}-t)})P_{b}(t,T)dZ_{t,b}^{M(T_{f})}
\]
\[
d\left(\frac{1}{P_{b}(t,T_{f}+\tau)}\right)=(..).dt+\frac{\sigma_{b}}{a}(1-e^{-a(T_{f}+\tau-t)}) \frac{1}{P_{b}(t,T_{f}+\tau)}dZ_{t.b}^{M(T_{f}+\tau)}
\]
The volatility of $V(t)$ adds basis contributions
\begin{align*}
dV_{t} =&(..)dt + (..)dZ_{t}^{M(T)} + (..) dZ_{t}^{M(T_{f})} + (..) dZ_{t}^{M(T_{f}+\tau)}\\
 & {}-\frac{\sigma_{b}}{a}(1-e^{-a(T_{f}-t)})  P(t,T)  N (1+\delta_{\tau}L(t,T_{f},T_{f}+\tau))  dZ_{t,b}^{M(T_{f})}\\
 & {}+\frac{\sigma_{b}}{a}(1-e^{-a(T_{f}+\tau-t)})  P(t,T)  N (1+\delta_{\tau}L(t,T_{f},T_{f}+\tau))  dZ_{t.b}^{M(T_{f}+\tau)}
\end{align*}
This gives additional contributions to two buckets, one for the fixing date
and another for the fixing date shifted by the underlying tenor:
\begin{enumerate}
\setcounter{enumi}{3}
\item fixing date:
\begin{itemize}
\item Maturity bucket $M(T_{f})$
\item Effective Notional: $N  P(0,T) (1+\delta_{\tau}  L(0,T_{f},T_{f}+\tau))$
\item Duration: $\frac{1}{a} (1-e^{-aT_{f}})$
\item Delta : $-1$
\end{itemize}
\item payment date of the underlying money market index :
\begin{itemize}
\item Maturity bucket $M(T_{f}+\tau)$
\item Effective Notional: $N  P(0,T) (1+\tau  L(0,T_{f},T_{f}+\tau))$
\item Supervisory Duration: $\frac{1}{a} (1-e^{-a (T_{f}+\tau)})$
\item Delta : $1$
\end{itemize}
\end{enumerate}

%=========================================================
\section{\sacB performance versus \sac}\label{s:numerical}

Here we compare \sacB \addons versus \sac \addons, and versus a simulation  of \gmm, and provide reference comparison to a 1-Factor Hull-White for completeness.  Both \sacB and \sac can be considered approximations to the simulation results.  These examples demonstrate that \sacB has appropriate risk sensitivity as compared to \sac, and showing results of ambiguity resolution in \sacB, i.e. it is clear what to do for zero coupon swaps.  Of course, by construction \sacB does not suffer from \regarb.
\begin{itemize}
\item Vanilla Swaps with different moneyness, Table \ref{fig:moneyness}.  \sac is insensitive to moneyness, showing a 56\%\ range of error w.r.t. the simulation, whereas \sacB error range is 2\%.
\begin{table}[h]
\begin{adjustwidth}{-2cm}{-2cm}
	\centering
	\begin{tabular}{p{2.6cm}p{1.9cm}p{1.9cm}p{1.9cm}p{1.9cm}p{1.9cm}}
\toprule
  Instrument & SA-CCR & RSA-CCR with no discounting & RSA-CCR with market discount & HW1F Average Pay/Receive Add-ons & Shifted LMM 3F Average Pay/Receive Add-on \\
\midrule
    ATM Swap &     4\% &                          3\% &                          -5\% &                               0\% &                                 3,783,285 \\
 ATM+100 bps &    -1\% &                          3\% &                          -5\% &                               0\% &                                 3,984,416 \\
  ATM+500bps &   -19\% &                          1\% &                          -6\% &                               0\% &                                 4,829,860 \\
  ATM-100bps &    10\% &                          3\% &                          -4\% &                              -0\% &                                 3,588,130 \\
  ATM-500bps &    37\% &                          1\% &                          -6\% &                              -0\% &                                 2,875,001 \\
\bottomrule
\end{tabular}

	\caption{Vanilla Swaps with varying moneyness, percentage differences from \gmm in last column.  The ATM swap is referred to below as \AS.}
	\label{fig:moneyness}
\end{adjustwidth}
\end{table}

\item Vanilla swaps with different replications, Table \ref{fig:replications}.  Here \sac's maximum error w.r.t. simulation is infinite as it gives non-zero risk for a zero economic position.  In contrast \sacB has a 5\%\ range or errors, and maximum error of -5\%.
\begin{table}[h]
\begin{adjustwidth}{-2cm}{-2cm}
	\centering
	\begin{tabular}{p{2.6cm}p{1.9cm}p{1.9cm}p{1.9cm}p{1.9cm}p{1.9cm}}
\toprule
      Instrument & SA-CCR & RSA-CCR with no discounting & RSA-CCR with market discount & HW1F Average Pay/Receive Add-ons & Shifted LMM 3F Average Pay/Receive Add-on \\
\midrule
        ATM Swap &     4\% &                          3\% &                          -5\% &                               0\% &                                 3,783,285 \\
 FRA Replication &    -9\% &                          3\% &                          -5\% &                               0\% &                                 3,783,285 \\
     ATM -- FRAs &   inf\% &                          0\% &                           0\% &                               0\% &                                         0 \\
     split at 3Y &    -3\% &                          3\% &                          -5\% &                               0\% &                                 3,783,285 \\
\bottomrule
\end{tabular}

	\caption{Vanilla Swaps with different replications, percentage differences from HW1F in last column.}
	\label{fig:replications}
\end{adjustwidth}
\end{table}

\item Amortising/Accreting vanilla swaps vs combination of swaps,  Table \ref{fig:amortising}.  Here the \sac range of error is 98\%\ (in one case it gives zero \addon for non-zero risk).  The range of error for \sacB is 28\%\ with no discounting or 14\%\ with market discounting versus the simulation.
\begin{table}[h]
\begin{adjustwidth}{-2cm}{-2cm}
	\centering
	\begin{tabular}{p{5.2cm}p{1.9cm}p{1.9cm}p{1.9cm}p{1.9cm}p{1.9cm}}
\toprule
                          Instrument & SA-CCR &  Full approach with discount equal 1 & RSA-CCR with market discount & HW1F Average Pay/Receive Add-ons & Shifted LMM 3F Average Pay/Receive Add-on \\
\midrule
      Amortising 10Y:5Y/100M;5Y/200M &    -2\% &                                   1\% &                          -5\% &                              -2\% &                                 5,996,321 \\
                     FRA replication &    -8\% &                                   1\% &                          -5\% &                              -2\% &                                 5,996,321 \\
 5Y/200M and forward start 5Y5Y/100M &    -4\% &                                   1\% &                          -5\% &                              -2\% &                                 5,996,321 \\
                10Y/100M and 5Y/100M &    -5\% &                                   1\% &                          -5\% &                              -2\% &                                 5,996,321 \\
           Amortising minus 10Y/150M &  -100\% &                                 -29\% &                         -17\% &                             -25\% &                                   352,986 \\
\bottomrule
\end{tabular}

	\caption{Amortising Swaps, percentage differences from \gmm in last column.}
	\label{fig:amortising}
\end{adjustwidth}
\end{table}

\item Zero coupon swaps,  Table \ref{fig:zero coupons}.  The range of error of \sac is 99\%\ in contrast to \sacB with 3\%.
\begin{table}[h]
\begin{adjustwidth}{-2.3cm}{-2.3cm}
	\centering
	\begin{tabular}{p{5.2cm}p{1.9cm}p{1.9cm}p{1.9cm}p{1.9cm}p{1.9cm}}
\toprule
                   Instrument & SA-CCR & Full approach with discount equal 1  & RSA-CCR with market discount & HW1F Average Pay/Receive Add-ons & Shifted LMM 3F Average Pay/Receive Add-on \\
\midrule
         Zero Coupon ATM  10y &    -2\% &                                  15\% &                          -3\% &                              -1\% &                                 3,999,402 \\
 Zero Coupon ATM+100bps 10 yr &    -9\% &                                  15\% &                          -3\% &                              -1\% &                                 4,341,650 \\
 Zero Coupon ATM+500bps 10 yr &   -27\% &                                  15\% &                          -2\% &                              -3\% &                                 5,368,783 \\
 Zero Coupon ATM-100bps 10 yr &     8\% &                                  15\% &                          -3\% &                              -1\% &                                 3,657,255 \\
 Zero Coupon ATM-500bps 10 yr &    72\% &                                  14\% &                          -3\% &                              -4\% &                                 2,290,554 \\
\bottomrule
\end{tabular}

	\caption{Zero Coupon Swaps, percentage differences from \gmm in last column.}
	\label{fig:zero coupons}
\end{adjustwidth}
\end{table}

\item Forward starts,  Table \ref{fig:forward start swaps}.  The range of error of \sac is 68\%\ in contrast to \sacB with 1\%.
\begin{table}[h]
\begin{adjustwidth}{-2cm}{-2cm}
	\centering
	\begin{tabular}{p{5.2cm}p{1.9cm}p{1.9cm}p{1.9cm}p{1.9cm}p{1.9cm}}
\toprule
                   Instrument & SA-CCR & RSA-CCR with no discounting & RSA-CCR with market discount & HW1F Average Pay/Receive Add-ons & Shifted LMM 3F Average Pay/Receive Add-on \\
\midrule
         Fwd Start ATM 1y-10y &     2\% &                          6\% &                          -3\% &                               1\% &                                 3,668,070 \\
 Fwd Start ATM 1y-10y +100bps &    -4\% &                          6\% &                          -3\% &                               1\% &                                 3,888,496 \\
 Fwd Start ATM 1y-10y +500bps &   -22\% &                          7\% &                          -3\% &                               1\% &                                 4,772,075 \\
 Fwd Start ATM 1y-10y -100bps &     9\% &                          6\% &                          -3\% &                               1\% &                                 3,447,875 \\
 Fwd Start ATM 1y-10y -500bps &    46\% &                          6\% &                          -2\% &                               2\% &                                 2,570,157 \\
\bottomrule
\end{tabular}

	\caption{Forward Staring Swaps, percentage differences from \gmm in last column.}
	\label{fig:forward start swaps}
\end{adjustwidth}
\end{table}

\end{itemize}

\FloatBarrier
%=========================================================
\section{Conclusions}

\sac has major issues including: lack of self-consistency for linear trades; lack of appropriate risk sensitivity (zero positions can have material \addons; moneyness is ignored); \regarb.   We have shown that \sac is, by parameter identification and re-construction, based on a 3-factor \hw model.    Hence we propose \sacB based on cashflow decomposition and this 3-factor \hw model.   \sacB is both free of \sac's issues, simple to use in practice, and can be extended easily given that it is model-based.    We recommend updating \sac to \sacB in order to resolve \sac's issues of lack of self-consistency for linear trades, lack of appropriate risk sensitivity (zero positions can have material \addons; moneyness is ignored), \regarb, and ambiguity of application for cases not explicitly described.

For a consistent treatment of the FX risk factors within the RSA-CCR cash-flow decomposition approach, it is enough to consider only the currency pairs with respect to the domestic currency i.e. CCY/USD as hedging sets. The add-on contributions to the FX asset class (hedging set CCY/USD) for an elementary linear cashflow in the currency CCY (regardless of the product type) is simply the present value of the cash-flow converted to the domestic currency USD. Moreover, each cashflow has an interest rate add-on contribution as computed for each elementary cashflow by the RSA-CCR approach.  

We recommend updating \sac to \sacB in order to resolve \sac's issues of lack of self-consistency for linear trades, lack of appropriate risk sensitivity (zero positions can have material \addons; moneyness is ignored), \regarb, and ambiguity of application for cases not explicitly described.

%=========================================================
\appendix
\section{Parametrization of \gmm simulated using shifted libor market model}

The \gmm is simulated using a 3-Factor Shifted Libor Market Model where the calibration is derived from the \gmm setup.  Volatility and shift calibration is shown in Table \ref{t:calib}.

\begin{table}
\begin{center}
\includegraphics[width=1.00\textwidth,trim=0cm 5cm 0cm 0cm,clip]{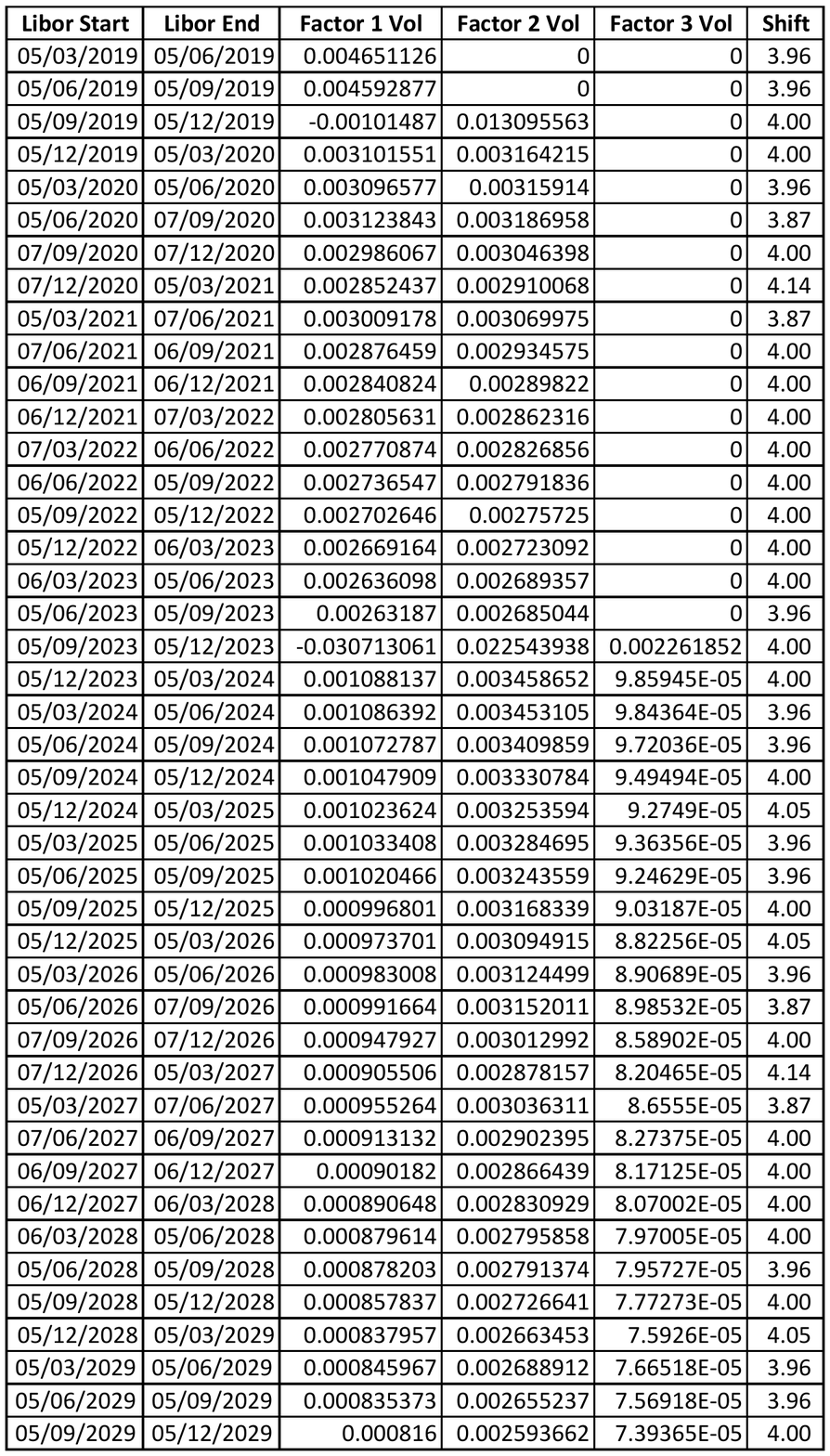}
\end{center}
\caption{ 3-Factor Shifted Libor Market Model calibration modeling \gmm.\label{t:calib}}
\end{table}

%=========================================================
\section{Acknowledgements}

The authors would like to gratefully acknowledge useful discussions with Lee McGinty and Crina Manolescu.

\bibliographystyle{chicago}
\bibliography{bibSACCR}

\end{document}